\documentclass[11pt,a4paper]{article}

\usepackage{enumitem}
\usepackage{xcolor}
\usepackage[utf8]{inputenc}
\usepackage{mathtools}
\usepackage{amsfonts}
\usepackage{amssymb}
\usepackage{amsthm}
\usepackage{tikz}
\usetikzlibrary{arrows,automata}
\usepackage{a4wide}

\newtheorem{theorem}{Theorem}
\newtheorem{definition}[theorem]{Definition}
\newtheorem{lemma}[theorem]{Lemma}
\newtheorem{proposition}[theorem]{Proposition}
\newtheorem{corollary}[theorem]{Corollary}
\newtheorem{example}[theorem]{Example}

\newcommand{\Nat}{{\mathbb{N}}}
\newcommand{\Real}{{\mathbb{R}}}
\newcommand{\blocks}[2]{\mathcal{B}_{#1}({#2})}
\newcommand{\tuple}[1]{\langle #1 \rangle}
\newcommand{\limrun}[0]{\infty}

\newcommand{\trans}[1]{\mathchoice{\xrightarrow{#1}}{\xrightarrow{\smash{\lower1pt\hbox{$\scriptstyle #1$}}}}{\text{Error}}{\text{Error}}}

\newcommand{\alocc}[2]{|\!|#1|\!|_{#2}}
\newcommand{\occ}[2]{|#1|_{#2}}

\newcommand{\bad}[0]{\operatorname{Bad}}

\begin{document}

\title{On normality in shifts of finite type}
\author{Nicol\'as \'Alvarez \and Olivier Carton}

\date{\today}
\maketitle

\begin{abstract}
  In this paper we consider the notion of normality of sequences in shifts
  of finite type.  A sequence is normal if the frequency of each
  block exists and is equal to the Parry measure of the block.  We give a
  characterization of normality in terms of incompressibility by lossless
  transducers.  The result was already known in the case of the full shift.
\end{abstract}


\section{Introduction}

In this paper, we extend to the context of shifts the strong link between
normality and incompressibility by finite state machines.  This link was
known for the full shift, that is the set of all infinite sequences over a
fixed alphabet since the work of Schnorr and Stimm \cite{SchnorrStimm71}
and Dai \textit{et al.} \cite{Dai04}.

On the one hand, normality is a weak notion of randomness.  It has been
introduced by Borel in~\cite{Borel09} more than a hundred years ago.  Roughly
speaking, it is required for a sequence to be normal that for each length,
all possible blocks of that length occur with the same frequency in the
sequence.  It has been shown by Borel that almost all sequences (in a
measure-theoretic sense) are normal but almost nothing is known about
specific sequences coming from number theory like the expansions in some
base of fundamental numbers as $\sqrt 2$ or~$\pi$.  On the other hand,
compressibility of sequences, especially by finite state machines, also
known as transducers, has been studied since the early days of computer
science \cite{Lempel78}.  This is mainly due to the large range of
applications of compression techniques.  These two important notions are
linked together by the characterization of normality by incompressibility.
Normal sequences are exactly those which cannot be compressed by some
one-to-one transducer.  This is a rather robust characterization as it is
valid for many variants of transducers: non-deterministic, two-way
\cite{BecherCartonHeiber15,CartonHeiber15}.

The notion of normality has been extended to broader contexts like the one
of dynamical systems and especially shifts of finite type \cite{Madritsch}.
These extensions are based on the different characterizations of normality,
namely uniform distribution \cite{Bertrand96} and martingales
\cite{AlmarzaFigueira15}.  When sofic shifts are irreducible and aperiodic,
they have a measure of maximal entropy and a sequence is then said to be
normal if the frequency of each block equals its measure.  This extension
to shifts meets the original aim of normality to study expansions of
numbers in bases when the shift arises from a numerical systems like the
$\beta$-shifts coming from the numeration in a non-integer base~$\beta$.
Normality can be again interpreted as the good distribution of blocks
of digits in the expansion of a number in a base~$\beta$.

There are two main contributions in this paper.  The first one is to give
different formulations of the definition of normality and to show that they
indeed coincide.  These formulations are similar to the ones used in the
classical setting of the full shift and their equivalence in that case is
well-known \cite{Bugeaud12}.  The second contribution is a characterization
theorem that proves that a sequence is normal in a shift of finite type
exactly when it is incompressible in the shift by finite state transducers.
Again this characterization holds in the case of the full shift
\cite{SchnorrStimm71,Dai04}.

This link between normality and automata opens the question whether
selecting digits in a sequence with an oblivious automaton preserves
normality in a shift of finite type as it does in the case of the full
shift \cite{Agafonov68}.  Oblivious means here that the selection of a
digit is based on the prefix of the sequence before the digit but not
including it.
 
The paper is organized as follows. Section~\ref{sec:prelim} introduces
all basic notions like shifts and normality. Section~\ref{sec:equiv}
is devoted to the equivalence of the different definitions of normality
given in the previous section.  The notion of incompressibility by
transducer is defined in Section~\ref{sec:compress}.  The main result is
stated and proved in Section~\ref{sec:main}.

\section{Preliminaries} \label{sec:prelim}

\subsection{Notation}

We write $\Nat$ for the set of all natural numbers.  An alphabet~$A$ is a
finite set with at least two symbols.  We write $A^\omega$ for the set of
all infinite words over~$A$ and $A^k$ stands for the set of all words of
length~$k$.  The length of a finite word $w$ is denoted by~$|w|$.  The
positions in finite and infinite words are numbered starting from~$1$.  For
a word~$w$ and positions $1 \le i \le j \le |w|$, we let $w[i]$ and
$w[i..j]$ denote respectively the symbol at position~$i$ and the subword
of~$w$ from position $i$ to position~$j$ (inclusive).

For any finite set $S$ we denote its cardinality with $|S|$.
We write $\log$ for the logarithm in base~$2$. 

\subsection{Shift spaces and subshifts of finite type}

In this article we are going to work on shift spaces, in particular
subshifts of finite type (SFT).  Let $A$ be a given alphabet.  The
\emph{full shift} is the set $A^\omega$ of all (one-sided) infinite
sequences $(x_n)_{n\ge 0}$ of symbols in~$A$.  The shift~$\sigma$ is the
function from~$A^\omega$ to~$A^\omega$ which maps each sequence
$(x_n)_{n\ge 0}$ to the sequence $(x_n)_{n\ge 1}$ obtained by removing the
first symbol.

Let $F \subset A^*$ a set of finite words called \emph{forbidden blocks}.
The subshift~$X_F$ is the subset of~$A^\omega$ made of sequences without any
occurrences of blocks in~$F$.  More formally, it is the set
\begin{displaymath}
  X_F = \{ x : x[m..n] \notin F \text{ for each } 1 \le m \le n\}.
\end{displaymath}

A \emph{shift space} of~$A^\omega$ or simply a \emph{shift} is a subset~$X$
of~$A^\omega$ which is closed for the product topology and invariant under
the shift operator, that is $\sigma(X) = X$.  This is equivalent to the
existence of a subset~$F \subset A^*$ of forbidden blocks such that
$X = X_F$.  The shift space is said to be of \emph{finite type} if
$X = X_F$ for some finite set~$F$ of forbidden blocks
\cite[Def.~2.1.1]{LindMarcus92}.  Up to a change of alphabet, any shift
space of finite type is the same as a shift space~$X_F$ where any forbidden
block has length~$2$, that is $F \subset A^2$.  For simplicity, we always
assume that each forbidden block has length~$2$.  In that case, the
set~$F$ is fully determined by the $A \times A$-matrix
$M = (m_{ab})_{a,b \in A}$ where $m_{ab} = 1$ if $ab \notin F$ and
$m_{ab} = 0$ otherwise and we write $X = X_M$.  The shift~$X$ is called
\emph{irreducible} if the graph induced by the matrix~$M$ is strongly
connected, that is, for each symbols $a,b \in A$, there exists an
integer~$n$ (depending on $a$ and~$b$) such that $M^n_{ab} > 0$.  The
shift~$X$ is called \emph{irreducible and aperiodic} if there exists an
integer~$n$ such that $M^n_{ab} > 0$ for each symbols $a,b \in A$.

For a shift~$X$ and $n \in \Nat$, we let $\blocks{}{X}$ (resp.,
$\blocks{n}{X}$ denote the set of all blocks (resp., blocks of length~$n$)
that occur in sequences of~$X$.  The \emph{(topological) entropy}~$h(X)$ of
the shift~$X$ is defined by
\begin{displaymath}
  h(X) = \lim_{n \to \infty} \frac{\log |\blocks{n}{X}|}{n}.
\end{displaymath}
If $X = X_M$ for some $\{0,1\}$ matrix~$M$, the entropy entropy $h(X)$ can
be computed as follows.  By the Perron-Frobenius theory, the non-negative
matrix~$M$ has a positive eigenvalue $\lambda$ of greatest modulus
\cite[Thm~1.5]{Senata06}.  The entropy~$h(X)$ of $X = X_M$ is then equal to
$\log \lambda$ \cite[Obs.~1.4.2]{Kitchens98}.

\begin{example}[Golden mean shift]
  The \emph{golden mean shift} is the shift space $X_F \subset
  \{0,1\}^\omega$ where the set of forbidden blocks is $F = \{ 11 \}$.  It
  is made of all sequences over $\{0,1\}$ with no two consecutive~$1$.
  This subshift is also equal to $X_M$ where $M = \left(\begin{smallmatrix}
      1&1 \\ 1&0
    \end{smallmatrix}\right)$ and its entropy is therefore $\log \lambda$
  where $\lambda = (1+\sqrt 5)/2$ is the golden mean.
\end{example}

A \emph{probability measure on $A^*$} is a function
$\mu : A^* \rightarrow [0,1]$ such that $\mu(\varepsilon) = 1$ and
\begin{displaymath}
  \sum_{a\in A}{\mu(wa)} = \mu(w)
\end{displaymath}
holds for each word $w \in A^*$.  The simplest example of a probability
measure is a \emph{Bernoulli measure}.  It is a morphism from~$A^*$
to~$[0,1]$ (endowed with multiplication) such that
$\sum_{a \in A}{\mu(a)} = 1$.  Among the Bernoulli measures there is the
\emph{uniform measure} which maps each word $w \in A^*$ to $|A|^{-|w|}$.
In particular, each symbol~$a$ is mapped to $\mu(a) = 1/|A|$.

By the Carath\'eodory extension theorem, a measure~$\mu$ on~$A^*$ can be
uniquely extended to a probability measure~$\hat{\mu}$ on~$A^\omega$ such
that $\hat{\mu}(wA^\omega) = \mu(w)$ holds for each word $w \in A^*$.  In
the rest of the paper, we use the same symbol for $\mu$ and~$\hat{\mu}$.  A
probability measure~$\mu$ is said to be \emph{(shift) invariant} if the
equality
\begin{displaymath}
  \sum_{a \in A}{\mu(aw)} = \mu(w)
\end{displaymath}
holds for each word $w \in A^*$.  We now introduce the entropy of a measure
\cite[Chap.~4]{Walter82}.  The \emph{entropy}~$h(\mu)$ of a measure~$\mu$
is defined by
\begin{displaymath}
  h(\mu) = \lim_{n \to \infty} -\frac{1}{n}\sum_{w \in A^n} \mu(w) \log \mu(w)
\end{displaymath}
with the usual convention $0 \log 0 = 0$.

For a stochastic matrix~$P$ and a stationary distribution~$\pi$, that is a
line vector such that $\pi P = \pi$, the \emph{Markov measure}
$\mu_{\pi,P}$ is the invariant measure defined by the following formula
\cite[Lemma~6.2.1]{Kitchens98}.
\begin{displaymath}
  \mu_{\pi,P} (a_1 a_2 \cdots a_k) = \pi_{a_1} P_{a_1a_2} \cdots P_{a_{k-1} a_k}
\end{displaymath}

A simple computation shows that the entropy $h(\mu_{\pi,P})$ of such a
measure is given by the following formula \cite[Obs.~6.2.10]{Kitchens98}.
\begin{displaymath}
  h(\mu_{\pi,P}) = -\sum_{i,j \in A} \pi_i P_{ij} \log P_{ij}
\end{displaymath}
with the convention $0 \log 0 = 0$.

A measure~$\mu$ is \emph{compatible} with a shift~$X$ if it only puts
weight on blocks of~$X$, that is, $\mu(w) > 0$ implies $w \in \blocks{}{X}$
for each word~$w$.  It is well known that each compatible measure~$\mu$
satisfies $h(\mu) \le h(X)$ \cite[Obs.~6.2.13]{Kitchens98}.  For a subshift
of finite type, there is a unique compatible measure with maximal
entropy~$h(X)$ \cite[Thm.~6.2.20]{Kitchens98}. This measure is called the
\emph{Parry measure} and it is a Markov measure. In the rest of the
document we let $\mu^X$ denote the Parry measure of an SFT~$X$.  This
measure can be explicitly given as follows.  The Parry measure of an
SFT~$X_M$ is the (one step) Markov measure given by the stochastic matrix
$P = (P_{i,j})$ where $P_{i,j} = M_{i,j} r_j/\lambda r_i$ and the
stationary probability distribution~$\pi$ defined by $\pi_i = l_ir_i$,
where $\lambda$ is the Perron eigenvalue of the matrix~$M$ and the vectors
$l$ and~$r$ are respectively left and right eigenvectors for~$\lambda$
normalized so that $\sum_{i=1}^k l_ir_i = 1$.

\begin{example}[Parry measure of the golden mean shift]
  Consider again the golden mean shift~$X$.  Its Parry measure is the
  Markov measure $\mu_{\pi,P}$ where $\pi$ is the distribution
  $\pi = (\lambda^2/(1+\lambda^2), 1/(1+\lambda^2))$ and $P$ is the
  stochastic matrix
  $P = \left(\begin{smallmatrix} 1/\lambda & 1/\lambda^2 \\
      1 & 0 \end{smallmatrix}\right)$ where $\lambda$ is the golden mean.
\end{example}

\subsection{Normality}

We start with the notation for the number of occurrences of a given word
within another word.

\begin{definition}[Occurences]
  For $w$ and $u$ two words, the number $\occ{w}{u}$ of \emph{occurrences}
  of~$u$ in~$w$ and the number $\alocc{w}{u,r}$ of \emph{aligned
    occurrences with offset $r$} of~$u$ in~$w$ are respectively given by
  \begin{align*}
    \occ{w}{u}  & =|\{ i : w[i..i+|u|-1] = u \}|, \\
    \alocc{w}{u,r} & =|\{ i : w[i..i+|u|-1] = u \text{ and } i = r \bmod |u|\}|.
  \end{align*}
  The number $\alocc{w}{u}$ of \emph{aligned occurrences} is given by
  \begin{align*}
  \alocc{w}{u} = \alocc{w}{u,1}
  \end{align*}
  
\end{definition}

For example, $\occ{aaaa}{aa} = 3$, $\alocc{aaaa}{aa} = 2$ and
$\alocc{aaaa}{aa,2} = 1$.

Borel's definition~\cite{Borel09} of normality for a sequence
$x \in A^\omega$ is that $x$ is \emph{normal} if for each
integer~$\ell \ge 1$ and each word $w \in A^{\ell}$ of length~$\ell$,
\begin{displaymath}
  \lim_{n \to \infty} \frac{\alocc{x[1..n\ell]}{w}}{n} = |A|^{-\ell}
\end{displaymath}

This definition is extended to the case of an SFT by replacing the uniform
measure by the Parry measure of the SFT.  A sequence~$x$ of an SFT~$X$ is
called \emph{normal (in~$X$)} if for each integer~$\ell \ge 1$ and each
word $w \in A^{\ell}$ of length~$\ell$,
\begin{displaymath}
 \lim_{n \to \infty} \frac{\alocc{x[1..n\ell]}{w}}{n} = \mu^X(w) 
\end{displaymath}
where $\mu^X$ is the Parry measure of~$X$.  This definition is based on
aligned occurrences.  It will be seen in the next section that alternative
definitions based on non-aligned occurrences are actually equivalent.

\section{Equivalence between definitions of normality} \label{sec:equiv}

In the literature there are several definitions of normality of a sequence.
Some of them are based on aligned occurrences and some others are based on
non-aligned occurrences.  It is part of the folklore that all these
definitions are indeed equivalent.  For the classical normality, proofs can
be found in~\cite[Thms 4.2 and 4.5]{Bugeaud12}.  For completeness, we
provide here a proof for the case of Markov measure.

\begin{theorem} \label{thm:equivdefs}
  Let $\mu$ be a Markov measure on~$A^*$.  For each sequence~$x$, the following
  three statements are equivalent.
  \begin{enumerate}[label=(\arabic*),ref=(\arabic*),itemsep=0cm]
  \item \label{ali} \emph{Aligned normality:} 
    for each integer~$\ell$ and each word $w \in A^{\ell}$, 
    \begin{displaymath}
     \lim_{n \to \infty} \frac{\alocc{x[1..n\ell]}{w}}{n} = \mu(w)
    \end{displaymath}
  \item \label{sali} \emph{Strong aligned normality:} 
    for each $\ell, k \in \Nat$ and each word $w \in A^{\ell}$,
    \begin{displaymath}
      \lim_{n \to \infty} \frac{\alocc{\sigma^k(x)[1..n\ell]}{w}}{n} =
      \mu(w)
    \end{displaymath}
  \item \label{nali} \emph{Non-aligned normality:} 
    for each word $w \in A^*$,
    \begin{displaymath}
     \lim_{n \to \infty} \frac{\occ{x[1..n]}{w}}{n} = \mu(w).
    \end{displaymath}
  \end{enumerate}
\end{theorem}

Before proving the theorem, we state two very simple but useful lemmas.
The first lemma states that obtaining a proper upper or lower bound for
asymptotic frequencies of all words of a given length is sufficient to
prove that limiting frequencies will follow the expected measure.  The
proof follows directly from the equality $\sum_{w \in A^{\ell}} \mu(w) = 1$
for each integer $\ell \ge 0$.

\begin{lemma} \label{lem:limtrick} 
  Let $\mu$ be a probability measure and $\ell$ a fixed non-negative
  integer. 
  For each sequence $x \in A^\omega$, the following three statements are
  equivalent.
  \begin{enumerate}[label=(\arabic*),ref=(\arabic*)]
  \item $\lim_{n \to \infty} \alocc{x[1..n\ell]}{w}/n = \mu(w)$ for each
    $w \in A^{\ell}$. 
  \item $\limsup_{n \to \infty} \alocc{x[1..n\ell]}{w}/n \le \mu(w)$ for each
    $w \in A^{\ell}$. 
  \item $\liminf_{n \to \infty} \alocc{x[1..n\ell]}{w}/n \ge \mu(w)$ for each
    $w \in A^{\ell}$. 
  \end{enumerate}
\end{lemma}

The next lemma states that it is sufficient to look at lengths which are
multiples of a fixed integer~$k$.  The proof follows easily from the
observation that there are at most $k$ occurrences of~$w$ starting between
positions $kn$ and $k(n+1)$.

\begin{lemma} \label{lem:multiples}
  Let $k \in \Nat$ be a fixed positive integer. For each sequence~$x$ and each $\ell \in \Nat$ and each finite word~$w \in A^\ell$, the following three statements hold.
  \begin{enumerate}[label=(\arabic*),ref=(\arabic*)]
  \item $\liminf_{n \to \infty} \alocc{x[1..n\ell]}{w}/n = \liminf_{n \to \infty}
    \alocc{x[1..nk\ell]}{w}/(nk)$ 
  \item $\limsup_{n \to \infty} \alocc{x[1..n\ell]}{w}/n = \limsup_{n \to \infty}
    \alocc{x[1..nk\ell]}{w}/(nk)$ 
  \item $\lim_{n \to \infty} \alocc{x[1..n\ell]}{w}/n = \lim_{n \to \infty}
    \alocc{x[1..nk\ell]}{w}/(nk)$ if such limits exist.
    \label{lim}
  \end{enumerate}
\end{lemma}

Lemma~\ref{lem:limtrick} and Lemma~\ref{lem:multiples} are stated for aligned occurrences frequencies but they are also valid for occurrences frequencies.

\begin{proof}[Proof of Theorem~\ref{thm:equivdefs}]
  The equivalence between the three definitions of normality is proved as
  follows.  We successively show that \ref{ali} implies \ref{sali},
  \ref{sali} implies \ref{nali} and that \ref{nali} implies \ref{ali}.

  \paragraph{\ref{ali} implies \ref{sali}}
  It is sufficient to prove that if $x$ presents aligned normality then
  $\sigma(x)$ also presents aligned normality.

  For $w \in A^{\ell}$, $k \ge \ell$ and $1 \le i \le k - \ell + 1$ we
  define $B(k,w,i)$ as the set of words of length~$k$ which contains an
  occurrence of~$w$ at position~$i$, that is
  $B = \{ v \in A^k : v[i..i + |w| - 1] = w \}$.  Since the Markov
  measure~$\mu$ is invariant $\mu(B(k,w,i)) = \mu(w)$ for any $|w| \le k$
  and $1 \le i \le k-|w|+1$.

  For any $w \in A^{\ell}$ and $r \in \Nat$.
  \begin{align*}
    \liminf_{n \to \infty} \frac{\alocc{\sigma(x)[1..n \ell]}{w}}{n}
    & = \liminf_{n \to \infty} \frac{\alocc{\sigma(x)[1..nr\ell]}{w}}{nr} \\ 
    & \ge \liminf_{n \to \infty} \frac{1}{r}\ \sum_{k=0}^{r-2}
      \;\sum_{v \in B(r \ell, w, 2 + \ell k)}
      \frac{\alocc{x[1..n r \ell]}{v}}{n} \\
    & = \frac{1}{r}\sum_{k = 0}^{r-2}\sum_{v \in B(\ell r, w, 2+\ell k)}{\mu(v)} \\
    & = \frac{r-1}{r} \mu(w)
  \end{align*}

  Since this inequality holds for any $r \in \Nat$.
  \begin{displaymath}
   \liminf_{n \to \infty} \frac{\alocc{\sigma(x)[1..n \ell]}{w}}{n} 
   \ge \mu(w) 
  \end{displaymath}
  and we conclude by Lemma~\ref{lem:limtrick}.

  \paragraph{\ref{sali} implies \ref{nali}}
  Notice that for any $w \in A^{\ell}$,
  \begin{displaymath}
    \occ{x[1..n]}{w} = 
    \sum_{i = 0}^{\ell-1} \alocc{\sigma^i(x)[1..n-i]}{w}
  \end{displaymath}

  then
  \begin{align*}
    \lim_{n \to \infty} \frac{\occ{x[1..n]}{w}}{n} & = 
    \sum_{i = 0}^{\ell-1} \lim_{n \to \infty} 
                        \frac{\alocc{\sigma^i(x)[1..n-i]}{w}}{n} \\
    & = \sum_{i = 0}^{\ell-1} \mu(w) / \ell = \mu(w) 
  \end{align*}

  \paragraph{\ref{nali} implies \ref{ali}}
  Let $w$ be a finite word of length~$\ell$.  For each word~$v$, we define
  $\alocc{v}{w,*} = \max_{i=1}^{\ell} \alocc{v}{w,i}$.  And, for a given
  $\varepsilon > 0$ and $k \in \Nat$, we define a set
  $\bad(w,k,\varepsilon)$ of words of length~$k\ell-1$ where the frequency
  of aligned occurrences of~$w$ is bad:
  \begin{displaymath}
    \bad(w,k,\varepsilon) = 
    \{v \in A^{k\ell-1} : \alocc{v}{w,*} > (k-1)(\mu(w) + \varepsilon) \}.
  \end{displaymath}

  By the ergodic theorem for irreducible Markov chains
  \cite[Thm.~1.10.2]{Norris98}, for each positive real numbers
  $\delta,\varepsilon > 0$, there exists $k_0$ such that for any
  $k \ge k_0$,
  \begin{displaymath}
    \mu(\bad(w,k,\varepsilon)) < \delta.
  \end{displaymath}

  Now, for any such $k \ge k_0$,
  \begin{align*}
    \limsup_{n \to \infty} \frac{\alocc{x[1..n\ell]}{w}}{n} 
    & = \limsup_{n \to \infty} \frac{\alocc{x[(k-1)\ell+1..n\ell]}{w}}{n} \\
    & \le \limsup_{n \to \infty} \frac{1}{n(k-1)\ell} \sum_{t = 1}^{(n-1)\ell+1} 
         \alocc{x[t..t + k\ell-2]}{w,2-t} \\
    & \le \limsup_{n \to \infty} \frac{1}{n(k-1)\ell} \sum_{t = 1}^{(n-1)\ell+1} 
         \alocc{x[t..t + k\ell-2]}{w,*} \\
    & = \limsup_{n \to \infty} \sum_{v \in A^{k\ell-1}} 
          \frac{\occ{x[1..(n+k-1)\ell -1]}{v}}{n\ell}
          \frac{\alocc{v}{w,*}}{k-1} \\
    & \le \sum_{v \in A^{k\ell-1}} \left( \limsup_{n \to \infty} 
          \frac{\occ{x[1..(n+k-1)\ell - 1]}{v}}{n\ell}\right)
          \frac{\alocc{v}{w,*}}{k-1} \\
    & = \sum_{v \in A^{k\ell -1}} \left( \limsup_{n \to \infty} 
          \frac{\occ{x[1..n\ell]}{v}}{n\ell}\right)
          \frac{\alocc{v}{w,*}}{k-1} \\
    & = \sum_{v \in A^{k\ell-1}} \mu(v)\ \frac{\alocc{v}{w,*}}{k-1} \\
    & = \sum_{v \in A^{k\ell-1} \setminus \bad(w,k,\epsilon)} \mu(v)
          \frac{\alocc{v}{w,*}}{k-1} 
          + \sum_{v \in \bad(w,k,\epsilon)} \mu(v)\frac{\alocc{v}{w,*}}{k-1} \\
    & \le (\mu(w)+\varepsilon)
          \sum_{v \in A^{k\ell-1} \setminus \bad(w,k,\epsilon)} \mu(v) + 
          \sum_{v \in A^{k\ell-1} \setminus \bad(w,k,\epsilon)} \mu(v) \\
    & \le \mu(w) + \varepsilon + \delta 
  \end{align*}

  The inequality on the second line comes from the fact that every aligned
  occurrence of $w$ in a position $j\ell + 1$ with $k-1 \le j < n$ is
  counted $(k-1)\ell$ times as $\alocc{x[t..t + k\ell -2]}{w,2-t}$ for
  $(j+1-k)\ell + 2 \le t \le j\ell + 1$.  This technique is due to Cassels
  \cite{Cassels52}.  Since the last inequality is true for any
  $\delta, \varepsilon > 0$, it follows that
  $\limsup_{n \to \infty}{\alocc{x[1..n\ell]}{w}/n} \le \mu(w)$ and we
  conclude by Lemma~\ref{lem:limtrick}.
\end{proof}

\section{Finite-state compressibility} \label{sec:compress}

In this section, we introduce the automata with output also known as
transducers which are used to characterize normality by incompressibility.
We consider \emph{non-deterministic transducers} computing functions from
sequences in a shift~$X$ to sequences in a shift~$Y$, that is, for a given
input sequence $x \in X$, there is at most one output sequence $y \in Y$.
We focus on transducer that operate in real-time, that is, they process
exactly one input alphabet symbol per transition.  We start with the
definition of a transducer.

\begin{definition}
  A \emph{non-deterministic transducer} is a tuple $\mathcal{T} =
  \tuple{Q,A,B,\delta,I,F}$, where
  \begin{itemize}
  \item $Q$ is a finite set of \emph{states},
  \item $A$ and $B$ are the input and output alphabets, respectively,
  \item $\delta \subset Q \times A \times B^* \times Q$ is a finite
    \emph{transition} relation,
  \item $I \subseteq Q$ and $F \subseteq Q$ are the sets of \emph{initial}
    and \emph{final} states, respectively.
  \end{itemize}
\end{definition}

A transition of such a transducer is a tuple $\tuple{p,a,v,q}$ in $Q \times
A \times B^* \times Q$ which is written $p \trans{a|v} q$.  A finite
(respectively infinite) \emph{run} is a finite (respectively infinite)
sequence of consecutive transitions,
\begin{displaymath}
  q_0 \trans{a_1|v_1} q_1 \trans{a_2|v_2} q_2
  \cdots q_{n-1} \trans{a_n|v_n} q_n
  \quad(\text{resp.}\quad
  q_0 \trans{a_1|v_1} q_1 \trans{a_2|v_2} q_2
 \trans{a_3|v_3} q_3  \cdots).
\end{displaymath}
Its \emph{input and output labels} are $a_1\cdots a_n$ and $v_1 \cdots v_n$ 
respectively.  A finite run is written $q_0 \trans{a_1\cdots a_n|v_1 \cdots
  v_n} q_n$.  An infinite run is \emph{final} if the state~$q_n$ is final
for infinitely many integers~$n$.  In that case, the infinite run is
written $q_0 \trans{a_1a_2a_3\cdots|v_1v_2v_3\cdots} \limrun$.  An infinite
run is accepting if it is final and furthermore its first state~$q_0$ is
initial.  This is the classical B\"uchi acceptance condition
\cite{PerrinPin04}.  We always assume that for each sequence~$x$, there is
at most one sequence~$y$ such that there is an accepting run $q_0
\trans{x|y} \limrun$ and we write $y = \mathcal{T}(x)$.  In that case, it
can be assumed that there is exactly one accepting run with input
label~$x$.  By a slight abuse of notation, we write $\mathcal{T}(x[m..n])$
for the output of~$\mathcal{T}$ along that run while reading the factor
$x[m..n]$.  We always assume that all transducers are \emph{trim}: each
state can occur in an accepting run.

A transducer $\mathcal{T}$ is called \emph{bounded-to-one} (resp.,
\emph{one-to-one}) if there is a constant~$K$ such that for each
sequence~$y$ the set $\mathcal{T}^{-1}(y) = \{x : \mathcal{T}(x) = y\}$
has cardinality a most~$K$ (resp., at most $1$).  We call here
\emph{compressor} a bounded-to-one transducer. In the literature,
\emph{lossless} deterministic transducers are often considered. As it was
shown in \cite[Prop~2.1]{BecherCartonHeiber15}, this is an intermediate
notion between one-to-one and bounded-to-one.  We prefer not to use this
notion as it is a structural property of the transducer and not of the
function it realizes.

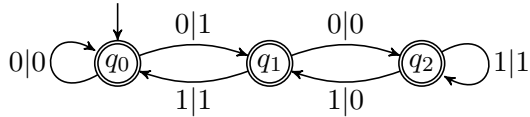
\begin{figure}[htbp]
  \begin{center}
  \begin{tikzpicture}[->,>=stealth',initial text=,semithick,auto,inner sep=1.5pt]
    \tikzstyle{every state}=[minimum size=0.4]
    \node[state,initial above,accepting] (q0) at (0,0) {$q_0$};
    \node[state,accepting]  (q1) at (2,0) {$q_1$};
    \node[state,accepting]  (q2) at (4,0) {$q_2$};
    \path (q0) edge[out=210,in=150,loop] node {$0|0$} ();
    \path (q0) edge[bend left=20] node {$0|1$} (q1);
    \path (q1) edge[bend left=20] node {$1|1$} (q0);
    \path (q1) edge[bend left=20] node {$0|0$} (q2);
    \path (q2) edge[out=30,in=-30,loop] node {$1|1$} ();
    \path (q2) edge[bend left=20] node {$1|0$} (q1);
  \end{tikzpicture}
  \end{center}
  \caption{A transducer for the multiplication by~$3$ in base~$2$}
  \label{fig:multby3}
\end{figure}
The transducer pictured in Figure~\ref{fig:multby3} is non-deterministic.
It realizes multiplication by~$3$ on binary expansions of real numbers.  If
the input~$x$ is the binary expansion of some real number $\alpha < 1/3$,
then the output is the binary expansion of~$3\alpha$.

The \emph{compression ratio} $\rho_{\mathcal{C}}(x)$ of a
compressor~$\mathcal{C}$ on a sequence~$x$ is
\begin{displaymath}
 \rho_{\mathcal{C}}(x) = \liminf_{n\to\infty}{\frac{|\mathcal{C}(x[1..n])|}{n}}.
\end{displaymath}
A sequence~$x$ of a shift~$X$ is called \emph{compressible in~$X$} if there
is a compressor $\mathcal{C}: X \to X$ such that
$\rho_{\mathcal{C}}(x) < 1$.

\section{Main Result} \label{sec:main}

It follows from the results in~\cite{SchnorrStimm71,Dai04} that the words
$x$ with compression ratio $\rho(x)$ equal to~$1$ are exactly the normal
words in the full shift.  A direct proof of this result appears
in~\cite{BecherHeiber13}. Extensions of this characterization for
non-determinism and extra memory appear
in~\cite{BecherCartonHeiber15,CartonHeiber15}.  The following theorem
extends this result to the context of shifts of finite type.

\begin{theorem} \label{thm:main}
  Let $X$ be an irreducible shift of finite type and $x$ a sequence in~$X$.
  The sequence~$x$ is normal in~$X$ if and only it is incompressible
  in~$X$.
\end{theorem}

The following proposition is a very classical result showing that elements
of a given shift can be encoded in another shift with a compression ratio
close to the ratio of their entropies.  To simplify the proof, we assume
that the shifts are aperiodic but this is not really necessary.  The result
can also be extended to sofic shifts.
\begin{proposition} \label{pro:coding} 
  Let $X$ and $Y$ be two irreducible and aperiodic shifts of finite type.
  For each real number $\varepsilon > 0$, there is a one-to-one
  transducer~$\mathcal{T}$ from~$X$ into~$Y$ such that for each $x \in X$,
  \begin{displaymath}
    \limsup_{n \to \infty} \frac{|\mathcal{T}(x[1..n]|)}{n} <
    \frac{h(X)}{h(Y)} + \varepsilon.
  \end{displaymath}
\end{proposition}
\begin{proof}
  Suppose that $X$ and $Y$ are the subshifts $X = X_M$ and $Y = X_N$
  for the two $\{0, 1\}$ matrices $M$ and~$N$.  Let $\lambda$ and~$\mu$ be
  respectively the eigenvalues of greatest modulus of $M$ and~$N$, so that
  $h(X) = \log \lambda$ and $h(Y) = \log \mu$. Let $p/q$ be a rational
  number such that $h(X)/h(Y) < p/q < h(X)/h(Y) + \varepsilon$.
  
  Since the sequence $M^n/\lambda^n$ converges to some matrix, there is a
  constant~$c$ such that $\sum_{a,b \in A}{M^n_{a,b}} \le c\lambda^n$ for
  each integer $n \ge 0$.  Similarly, there is another constant~$d$, a
  symbol~$a$ and an integer~$n_0$ such that $N^n_{a,a} \ge d\mu^n$ for each
  $n \ge n_0$.  Since $q\log \lambda < p\log \mu$, there is an
  integer~$n_1$ such that $c\lambda^{qn} < d\mu^{pn}$ for each $n \ge n_1$.
  Let us recall that $M^n_{a,b}$ is the number of words~$w$ of length~$n-1$
  such that $awb$ is a block of~$X$ and the sum
  $\sum_{a,b \in A}{M^n_{a,b}}$ is thus the number of blocks of length
  $n+1$ in~$X$.

  From the previous inequalities, it follows that, for
  $n \ge \max(n_0,n_1)$, the number $|\blocks{qn+1}{X}|$ of blocks of
  length~$qn+1$ in~$X$ is less than the number of words~$w$ of length~$pn-1$
  such that $awa$ is a block of~$Y$. Let us choose an integer~$n$ such that
  $n \ge \max(n_0,n_1)$.  Let
  $f$ be a one-to-one function which maps each block~$u$ of length~$qn+1$
  of~$X$ to a word $w = f(u)$ such that $awa$ is a block of~$Y$.  The
  transducer~$\mathcal{T}$ reads each word $x \in X$ by blocks of
  length~$qn+1$. For each read block~$u$ of length~$qn+1$, it outputs $aw$
  where $w = f(u)$.  Since $f$ is one-to-one, the function realized
  by~$\mathcal{T}$ is also one-to-one.  Furthermore, 
  $\lim_{N \to \infty} |\mathcal{T}(x[1..N]|)/N = pn/(qn+1) < h(X)/h(Y) + \varepsilon$.
\end{proof}

The following corollary allows us to work with transducers from a specific
shift into the full-shift and adapt those results to the case where the
transducer has the same shift as domain and image.
\begin{corollary} \label{cor:coding}
  There is a compressor $\mathcal{C}: X \to X$ such that
  $\rho_{\mathcal{C}}(x) < 1$ if and only if there is a compressor
  $\mathcal{C}': X \to 2^\omega$ such that $\rho_{\mathcal{C}'}(x) < h(X)$.
\end{corollary}
\begin{proof}
  Suppose that there is a compressor $\mathcal{C}: X \to X$ such that
  $\rho_{\mathcal{C}}(x) < 1$.  Let $\varepsilon$ be a positive real number
  such that $(h(X)+\varepsilon)\rho_{\mathcal{C}}(x) < h(X)$. By the
  previous proposition with $Y = 2^\omega$, there is a transducer~$\mathcal{T}$
  from~$X$ to~$2^\omega$ such that
  $\limsup_{n \to \infty} |\mathcal{T}(x[1..n]|)/n \le h(X) + \varepsilon$.
  The composition $\mathcal{C}' = \mathcal{T} \circ \mathcal{C}$ gives the
  required compressor.  The converse is proved similarly by using the
  previous proposition with $X = 2^\omega$ and $Y = X$.
\end{proof}

The following result is a classical generalization of Kraft's inequality.
It is the key lemma used to prove that normal sequences cannot be
compressed by finite state machines.
\begin{lemma} \label{lem:kraft}
  Let $\mathcal{C}$ be a compressor from~$X$ to~$Y$ with $|Q|$ states.  For
  each word~$w$, let $L_{\mathcal{C}}(w)$ be the minimum number of symbols
  written by a finite run in~$\mathcal{C}$ with input label~$w$.  Then
  \begin{displaymath}
    \sum_{w \in A^{\ell}} 2^{-L_{\mathcal{C}}(w)} \le
    K|Q|^2 (1 + \ell r_{\mathcal{C}})
  \end{displaymath}
  where $|\mathcal{C}^{-1}(y)| \le K$ for each $y \in Y$ and
  $r_{\mathcal{C}}$ is the maximum number of symbols written by a single
  transition of~$\mathcal{C}$.
\end{lemma}
\begin{proof}
  We claim that for each integer~$k$ the cardinality of the set
  $\{ w \in A^{\ell}: L_{\mathcal{C}}(w) = k \}$ is at most $K|Q|^22^k$.
  Let $p$ and $q$ be two states of~$\mathcal{C}$ and $v$ a word of
  length~$k$.  We claim that the set
  $\{ w \in A^{\ell} : p \trans{w|v} q \}$ has cardinality at most~$K$.
  Suppose that there are $n$ distinct words $w_1,\ldots,w_n$ in this set.
  Since the transducer is trim, there is a finite run $i \trans{u|v'} p$
  from a n initial state~$i$ and a final run $q \trans{x|y} \infty$.  It
  follows that $\mathcal{C}(uw_jx) = v'vy$ for each $1 \le j \le n$ and
  thus $n \le K$.  If a word~$w$ belongs to
  $\{ w \in A^{\ell}: L_{\mathcal{C}}(w) = k \}$, then there are two states
  $p,q$ and a word~$v$ of length~$k$ such that $w$ belongs to
  $\{ w \in A^{\ell} : p \trans{w|v} q \}$.  This proves the upper bound
  for the cardinality of $\{ w \in A^{\ell}: L_{\mathcal{C}}(w) = k \}$
  since there are $|Q|^2$ possible choices for $p$ and~$q$ and $2^k$
  possible choices for~$v$.
  \begin{align*}
    \sum_{w \in A^{\ell}} 2^{-L_{\mathcal{C}}(w)} & =
    \sum_{k = 0}^{\ell r_{\mathcal{C}}} |\{ w \in A^{\ell}: L_{\mathcal{C}}(w) = k \}| 2^{-k}  \\
    & \le \sum_{k = 0}^{\ell r_{\mathcal{C}}} K|Q|^2  = K|Q|^2 (1 + \ell r_{\mathcal{C}})
  \end{align*}
\end{proof}

Let $u \in A^n$ and $w \in A^{\ell}$ be two finite words of length $n$
and~$\ell$ and let $x$ be an infinite word.  First define the relative
frequency $P(w,u)$ by $P(w,u) = \ell\alocc{u}{w}/n$.  This is just the
number of aligned occurrences of~$w$ in~$u$ normalized by the factor
$\ell/n$ such that $\sum_{w\in A^{\ell}}{P(w,u)} = 1$.  The
\emph{$\ell$-block entropy} $h_{\ell}(u)$ of~$u$ is then defined
$h_{\ell}(u) = - \frac{1}{\ell} \sum_{w \in A^{\ell}} P(w,u) \log P(w,u)$.
This $\ell$-block entropy is extended to infinite words by setting
$h_{\ell}(x) = \liminf_{k \to \infty} h_{\ell}(x[1..k\ell])$.  The
\emph{block entropy} $h(x)$ of~$x$ is then defined by
$h(x) = \liminf_{\ell \to \infty} h_{\ell}(x)$

It should be noted that the block entropy~$h_{\ell}(x)$ has been
defined using aligned occurrences.  This is the same as~$\hat{H}_{\ell}$
in~\cite{Sheinwald94} but not the same as~$\hat{H}_{\ell}$
in~\cite{Lempel78} where entropy is defined using non-aligned occurrences.
Therefore the existence of the limit $\lim_{\ell \to \infty} h_{\ell}(x)$
does not follow from the results in~\cite{Lempel78} and $h(x)$ is defined
as $h(x) = \liminf_{\ell \to \infty} h_{\ell}(x)$.
\begin{lemma}[Proof of Theorem~3 in \cite{Lempel78}] \label{lem:lempel}
  Given an alphabet $A$ and a sequence $x \in A^\omega$. For any
  compressor $\mathcal{C}: A^\omega \to \{0,1\}^\omega$:
  \begin{displaymath}
   \rho_{\mathcal{C}}(x) \ge h(x) 
  \end{displaymath}
\end{lemma}
For completeness, we present the proof of this theorem as given
in~\cite{Sheinwald94}.
\begin{proof}
  Let us consider a bounded-to-one compressor $\mathcal{C}$ with $|Q|$
  states.  Suppose that for each $y \in \{0,1\}^\omega$,
  $|\mathcal{C}^{-1}(y)| \le K$.  For a word $w \in A^{\ell}$,
  $\mathcal{C}$ produces an output depending on its current state.  Let us
  denote as $L_{\mathcal{C}}(w)$ the length of the shortest output that
  $\mathcal{C}$ produces when reading $w$, where the minimum is taken over
  all possible finite runs with $w$ as input label.
  \begin{displaymath}
    \rho_{\mathcal{C}}(x[1..k\ell])
    \ge \frac{1}{\ell} \sum_{w \in A^l} P(x[1..k\ell], w) \cdot L_{\mathcal{C}}(w) 
  \end{displaymath}
  Then,
  \begin{align*} 
  h_{\ell}(x[1..k\ell]) - & \rho_{\mathcal{C}}(x[1..k\ell]) \le \\
   & \frac{1}{\ell} \sum_{w \in A^{\ell}} P(x[1..k\ell], w) 
     \log \left(\frac{2^{-L_{\mathcal{C}}(w)}}{P(x[1..k\ell],w)}\right) 
  \end{align*}
  By Jensen inequality applied to the $\log$ function,
  \begin{displaymath} 
    h_{\ell}(x[1..k\ell])  - \rho_{\mathcal{C}}(x[1..k\ell]) \le  
    \frac{1}{\ell} \log \left( \sum_{w \in A^{\ell}} 2^{-L_{\mathcal{C}}(w)} \right) 
  \end{displaymath}
  By the generalized Kraft's inequality of Lemma~\ref{lem:kraft},
  \begin{displaymath}
    h_{\ell}(x[1..k\ell]) - \rho_{\mathcal{C}}(x[1..k\ell]) \le
    \frac{1}{\ell} \log (K|Q|^2 (1 + \ell r_{\mathcal{C}}))
  \end{displaymath}
  and taking first the limit when $k \to \infty$ and then the limit when
  $\ell \to \infty$ yields the required inequality $h(x) \le
  \rho_{\mathcal{C}}(x)$. 
\end{proof}

We now come to the proof of the main theorem.
\begin{proof}[Proof of Theorem~\ref{thm:main}]
  Assume that the sequence $x$ is normal in the shift~$X$.  Let $\phi:
  [0,1] \to \Real$ be defined as $\phi(p) = -p \log p$ with the usual
  convention that $0 \log 0 = 0$.  Since $\phi$ is a continuous function
  and for every word $w \in A^{\ell}$, $\lim_{k \to \infty} P(w,
  x[1..k\ell]) = \mu(w)$
  \begin{displaymath}
    h_{\ell}(x) = \frac{1}{\ell} \sum_{w \in A^{\ell}} \phi(\mu(w)).
  \end{displaymath}
  Therefore $h(x) = \liminf_{\ell \to \infty} h_{\ell}(x) = h(\mu) = h(X)$.
  By Lemma~\ref{lem:lempel}, there is no compressor
  $\mathcal{C}': X \to \{0,1\}^\omega$ with a compression ratio better
  than~$h(X)$.  By Corollary~\ref{cor:coding} we conclude that there is no
  compressor $\mathcal{C}: X \to X$ such that $\rho_{\mathcal{C}}(x) < 1$.

  Now suppose that the sequence~$x$ is not normal.  By definition, there is
  a finite word $w_0 \in A^*$ such that either
  \begin{displaymath}
    \lim_{n \to \infty}  \frac{\alocc{x[1..n\ell]}{w_0}}{n} \neq \mu^X(w_0)
  \end{displaymath}
  or this limit does not exist where $\ell = |w_0|$ is the length
  of~$w_0$.

  It is possible to choose a subsequence of positions
  $1 \le n_1 < n_2 < n_3 < \cdots$ such that the ratio
  $\alocc{x[1..n_i\ell]}{w}/n_i$ converges for every $w \in A^{\ell}$ and
  such that the limit of this ratio is different from $\mu(w_0)$ for
  $w = w_0$.

  Let $M = |\blocks{\ell}{X}|$ be the number of blocks of length~$\ell$
  in~$X$ and let $B = \{1, 2, \ldots, M\}$ be an alphabet of
  cardinality~$M$. We can encode $x$ into a sequence $y \in B^\omega$ by
  taking aligned words of length $\ell$ in $x$ and representing them as a
  single symbol of~$B$ using a bijective mapping
  $f: \blocks{\ell}{X} \to B$.  The sequence~$y$ belongs to a subshift of
  finite type~$Y$ with entropy $h(Y) = \ell h(X)$.

  For every $b \in B$, the limit
  $\lim_{i \to \infty} \occ{y[1..n_i]}{b}/n_i$ does exist, and for
  $b_0 = f(w_0)$, it satisfies
  $\lim_{i \to \infty} \occ{y[1..n_i]}{b_0}/n_i \neq \mu^Y(b_0) = \mu^X(w_0)$.

  Let $n'_1, n'_2, \ldots$ be a subsequence of $n_1, n_2, \ldots$ such that
  the ratio $\occ{y[1..n_i']}{ab}/n_i'$ converges for each $a,b \in B$.
  Define the distribution vector $\pi = (\pi_a)_{a \in B}$ and the
  stochastic matrix $P = (P_{ab})_{a,b \in B}$ by
  \begin{displaymath}
    \pi_a  = \lim_{i \to \infty} \frac{\occ{y[1..n_i'}{a}}{n_i'}
    \quad\text{and}\quad
    P_{ab} = 
    \begin{cases}
      \displaystyle
      \lim_{i \to \infty} \frac{\occ{y[1..n_i']}{ab}}{\occ{y[1..n_i']}{a}} &
        \text{ if } \pi_a \neq 0 \\
      \displaystyle
      \frac{1}{M} & \text{otherwise.}
      \end{cases}
  \end{displaymath}

  The stochastic matrix $P$ is used to defined a measure~$\nu$ on~$A^*$
  by setting for each word $a_1a_2 \cdots a_n \in B^*$
  \begin{displaymath}
    \nu(a_1a_2\cdots a_n) = \frac{1}{M} \prod_{i=1}^{n-1}{P_{a_{i} a_{i+1}}}.
  \end{displaymath}
  with the convention that $\nu(a) = 1/M$ for each symbol $a \in B$.
  Note that this measure might be not invariant because the vector
  $(1,\ldots,1)$ might be not a left eigenvector of the matrix~$P$.

  Let $k$ be an integer to be fixed later. We construct an appropriate
  encoding of~$B^k$ based on the values of~$\nu$. Some care must be taken
  for words where $\nu$ takes the value~$0$.  Let
  $S = \{ u \in B^k : \nu(u) = 0\}$ be the subset of words of length~$k$
  mapped to~$0$ by~$\nu$ and $T = B^k \setminus S$ be its complement.  Note
  that if $u = a_1a_2\cdots a_k$ belongs to~$S$, there is then some index
  $1 \le i \le k-1$ such that $P_{a_i a_{i+1}} = 0$, which means that
  $\lim_{i \to \infty} \occ{y[1..n_i']}{a_ia_{i+1}}/n_i' = 0$, and in turn
  $\lim_{i \to \infty} \alocc{y[1..n_i'}{u}/n_i' = 0$.

  If $S$ is non-empty, define a one-to-one mapping
  \begin{displaymath}
   C_S : S \to \{0,1\}^L \text{  where  } L = \lceil \log |S| \rceil 
  \end{displaymath}

  For $T$, we define a prefix-free code 
  \begin{displaymath}
    C_T : T \to \{0,1\}^* \text{  such that  }
    |C_T(u)| = \left\lceil -\log \nu(u)\right\rceil 
  \end{displaymath}
  The existence of such a code is guaranteed by Kraft's inequality since
  $\sum_{u \in T} \nu(u) = 1$.   The functions $C_S$ and~$C_T$ are now
  used to define a unique function~$C_k: B^k \to \{0,1\}^*$ as follows.
  \begin{displaymath}
    C_k(u) = 
    \begin{cases}
      0C_S(u) & \text{if  $u \in S$} \\
      1C_T(u) & \text{if  $u \in T$}
    \end{cases}
  \end{displaymath}
  Since both functions $C_S$ and~$C_T$ are one-to-one, the function~$C_k$
  is also one-to-one.  This latter function is now used to define a
  transducer $\mathcal{C}: Y \to \{0,1\}^\omega$ which reads each sequence
  in~$Y$ by blocks of length~$k$ and for each read block $u \in B^k$
  outputs $C_k(u)$.  Since the function $C_k$ is one-to-one, the
  transducer~$\mathcal{C}$ is also one-to-one.  We now estimate its
  compression ratio~$\rho_{\mathcal{C}}(y)$ on the input~$y$.
  
  \begin{align*}
  \rho_{\mathcal{C}}(y)
  & = \liminf_{n \to \infty} \frac{|\mathcal{C}(y[1..n])|}{n} \\
  & \le \liminf_{i \to \infty} \frac{|\mathcal{C}(y[1..n_i'])|}{n_i'} \\
  & = \liminf_{i \to \infty} \frac{1}{n_i'} \sum_{u \in B^k}
        \alocc{y[1..n_i']}{u}\ |C_k(u)| \\
  & = \liminf_{i \to \infty} \frac{1}{n_i'} 
      \left( 
         \sum_{u \in S} \alocc{y[1..n_i']}{u} (L+1)
         + \sum_{u \in T} \alocc{y[1..n_i']}{u} (1 + |C_T(u)|)
      \right) \\
  & = \liminf_{i \to \infty}  \frac{1}{n_i'} \sum_{u \in T} \alocc{y[1..n_i']}{u} (1+|C_T(u)|) \\
  & = \liminf_{i \to \infty} \frac{1}{n_i'} \sum_{u \in T} \alocc{y[1..n_i']}{u} 
      \left( 1 + \left\lceil -\log \nu(u) \right\rceil \right) \\
  & \le \liminf_{i \to \infty} \frac{1}{n_i'} \sum_{u \in T} \alocc{y[1..n_i']}{u} 
      \left( 2 + \log \frac{M}{\prod_{j=1}^{k-1} P_{u_i u_{i+1}}} \right) \\
  & = \frac{(2 + \log M) \left\lfloor n_i'/k \right\rfloor}{n_i'}
  - \limsup_{i \to \infty} \frac{1}{n_i'} \sum_{u \in T} \alocc{y[1..n_i']}{u} 
      \sum_{j=1}^{k-1} \log (P_{u_i u_{i+1}}) \\ 
  & \le \frac{2 + \log M}{k} - \limsup_{i \to \infty} \frac{1}{n_i'} 
  \sum_{j = 1}^{n_i'-1} \log P_{y_i y_{i+1}} \\
  & = \frac{2 + \log M}{k} 
      - \limsup_{i \to \infty} \sum_{a,b \in B} \frac{\occ{y[1..n_i']}{ab}}{n_i'} \log P_{ab} \\
  & = \frac{2 + \log M}{k} - \sum_{a,b \in B} \pi_a P_{ab} \log P_{ab}
  \end{align*}

  Since the last inequality is valid for any $k \in \Nat$, and
  \begin{displaymath}
    - \sum_{a,b \in B} \pi_a P_{ab} \log P_{ab} =
    h(\mu_{\pi,P}) < h(\mu^Y) = h(Y) = \ell h(X) 
  \end{displaymath}
  We conclude that there is a compressor
  $\mathcal{C} : Y \to \{0,1\}^\omega$, such that
  $\rho_{\mathcal{C}}(y) < \ell h(X)$.  Now, define the compressor
  $\mathcal{C}' : X \to \{0,1\}^\omega$, which takes blocks of $\ell$
  symbols from the input, maps them into $B$ using the bisection
  $f : \blocks{\ell}{X} \to B$ and then simulates the
  transducer~$\mathcal{\mathcal{C}}$ to produce a binary output.  Its
  compression ratio on the input~$x$ is given by
  \begin{displaymath}
    \rho_{\mathcal{C}'}(x) = \rho_{\mathcal{C}}(y)/\ell  < h(X).
  \end{displaymath}
  This inequality implies, by Corollary~\ref{cor:coding}, that there is a
  compressor $\mathcal{C}'' : X \to X$, such that $\rho_{\mathcal{C}''}(x) < 1$.
\end{proof} 

\section{Outlook}

The main theorem (Theorem~\ref{thm:main}) is stated and proved for shifts
of finite type for simplicity.  We would like to provide some evidence that
the result can be generalized to the case of sofic shifts.  The Parry
measure of an irreducible sofic shift does exist and it is an hidden Markov
chain (see \cite[Thm~1]{Weiss77}, \cite[Thm~4]{Ficher75} and
\cite[p.~444]{LindMarcus92}).

In the proof of Theorem~\ref{thm:equivdefs}, the fact that the
measure~$\mu$ is a Markov chain is only used through the ergodic theorem.
Since this latter result also holds for hidden Markov chains,
Theorem~\ref{thm:equivdefs} can be lifted to hidden Markov chains.

The proof of Proposition~\ref{pro:coding} can be adapted to sofic shifts.
The rest of the proof of Theorem~\ref{thm:main} does not really use
the fact that the shift~$X$ is of finite type.  

\section*{Acknowlegments}

The authors would like to thank the anonymous referee for useful
suggestions and Ver\'onica Becher for very fruitful discussions.  Alvarez's
work was financed through a postgraduate scholarship from CONICET (National
Scientific and Technical Research Council of Argentina) and his stays at
France to collaborate with Carton were possible thanks to financial support
from Laboratoire International Associé Infinis, the Ministry of Education
of Argentina and the French agency Campus France.  Carton is member of the
Laboratoire International Associ\'e INFINIS.  He is partially supported by
the ECOS project PA17C04 and by the DeLTA project (ANR-16-CE40-0007).

\bibliographystyle{plain}
\bibliography{shifts}

\begin{thebibliography}{10}

\bibitem{Agafonov68}
V.~N. Agafonov.
\newblock Normal sequences and finite automata.
\newblock {\em Soviet Mathematics Doklady}, 9:324--325, 1968.

\bibitem{AlmarzaFigueira15}
J.~Almarza and S.~Figueira.
\newblock Normality in non-integer bases and polynomial time randomness.
\newblock {\em Journal of Computer and System Sciences}, 81:1059--1087, 2015.

\bibitem{BecherCartonHeiber15}
V.~Becher, O.~Carton, and P.~A. Heiber.
\newblock Normality and automata.
\newblock {\em Journal of Computer and System Sciences}, 81(8):1592--1613,
  2015.

\bibitem{BecherHeiber13}
V.~Becher and P.~A. Heiber.
\newblock Normal numbers and finite automata.
\newblock {\em Theoretical Computer Science}, 477:109--116, 2013.

\bibitem{Bertrand96}
A.~Bertrand-Mathis.
\newblock Nombres normaux.
\newblock {\em Journal de th{\'e}orie des nombres de {B}ordeaux}, 8:397--412,
  1996.

\bibitem{Borel09}
\'E. Borel.
\newblock Les probabilit\'{e}s d\'{e}nombrables et leurs applications
  arithm\'{e}tiques.
\newblock {\em Rendiconti del Circolo Matematico di Palermo}, 27:247--271,
  1909.

\bibitem{Bugeaud12}
Y.~Bugeaud.
\newblock {\em Distribution Modulo One and Diophantine Approximation}.
\newblock Cambridge Tracts in Mathematics. Cambridge University Press, 2012.

\bibitem{CartonHeiber15}
O.~Carton and P.~A. Heiber.
\newblock Normality and two-way automata.
\newblock {\em Information and Computation}, 241:264--276, 2015.

\bibitem{Cassels52}
J.~W.~S. Cassels.
\newblock On a paper of {N}iven and {Z}uckerman.
\newblock {\em Pacific Journal of Mathematics}, 2(4):555--557, 1952.

\bibitem{Dai04}
J.~Dai, J.~Lathrop, J.~Lutz, and E.~Mayordomo.
\newblock Finite-state dimension.
\newblock {\em Theoretical Computer Science}, 310:1--33, 2004.

\bibitem{Ficher75}
R.~Ficher.
\newblock Sofic systems and graphs.
\newblock {\em Monatshefte f{\"u}r {M}athematik}, 80:179--186, 1975.

\bibitem{Kitchens98}
B.~P. Kitchens.
\newblock {\em Symbolic Dynamics}.
\newblock Springer, 1998.

\bibitem{Lempel78}
A.~Lempel.
\newblock Compression of individual sequences via variable-rate coding.
\newblock {\em IEEE Transactions on Information Theory}, 24(5):530--536, 1978.

\bibitem{LindMarcus92}
D.~Lind and B.~Marcus.
\newblock {\em An Introduction to Symbolic Dynamics and Coding}.
\newblock Cambridge University Press, 1992.

\bibitem{Madritsch}
M.~Madritsch.
\newblock Normal numbers and symbolic dynamics.
\newblock In {\em Sequences}, chapter~8. Cambridge University Press, 2018.

\bibitem{Norris98}
J.R. Norris.
\newblock {\em Markov Chains}.
\newblock Cambridge Series in Statistical and Probabilistic Mathematics.
  Cambridge University Press, 1998.

\bibitem{PerrinPin04}
D.~Perrin and J.-{\'E}. Pin.
\newblock {\em Infinite Words}.
\newblock Elsevier, 2004.

\bibitem{SchnorrStimm71}
C.~P. Schnorr and H.~Stimm.
\newblock Endliche {A}utomaten und {Z}ufallsfolgen.
\newblock {\em Acta Informatica}, 1:345--359, 1972.

\bibitem{Senata06}
E.~Senata.
\newblock {\em Non-negative Matrices and Markov Chains}.
\newblock Springer, 2006.

\bibitem{Sheinwald94}
D.~Sheinwald.
\newblock On the {Z}iv-{L}empel proof and related topics.
\newblock {\em Proceedings of the IEEE}, 82(6):866--871, 1994.

\bibitem{Walter82}
P.~Walter.
\newblock {\em An introduction to Ergodic Theory}.
\newblock Spinger, 1982.

\bibitem{Weiss77}
B.~Weiss.
\newblock Subshifts of finite type and sofic systems.
\newblock {\em Monatshefte f{\"u}r Mathematik}, 77:462--474, 1977.

\end{thebibliography}

\begin{minipage}{\textwidth}
\noindent
Nicol\'as Alvarez\\
ICIC - Universidad Nacional del Sur, CONICET\\
Departamento de Ciencias en Ingenier\'ia de la Computaci\'on\\
\texttt{naa@cs.uns.edu.ar}\\
\medskip\\
Olivier Carton\\
Institut de Recherche en Informatique Fondamentale\\
Universit\'e Paris Diderot\\
\texttt{Olivier.Carton@irif.fr}
\end{minipage}

\end{document}